\documentclass[letterpaper, 10 pt, conference]{ieeeconf}  
\usepackage{gen_settings}

\IEEEoverridecommandlockouts                              

\title{\LARGE \bf
Test and Evaluation of Quadrupedal Walking Gaits \\
through Sim2Real Gap Quantification
}

\author{Prithvi Akella, Wyatt Ubellacker, and Aaron D. Ames$^{1}$
\thanks{*This work was supported by AFOSR (FA9550-19-1-0302) and Dow (\#227027AT)}
\thanks{$^{1}$All authors are with the California Institute of Technology
        {\texttt{\{pakella,wubellac,ames\}@caltech.edu}}}%
}

\begin{document}

\maketitle
\thispagestyle{empty}
\pagestyle{empty}

\begin{abstract}

In this letter, the authors propose a two-step approach to evaluate and verify a true system's capacity to satisfy its operational objective.  Specifically, whenever the system objective has a quantifiable measure of satisfaction, \textit{i.e.} a signal temporal logic specification, a barrier function, \textit{etc} - the authors develop two separate optimization problems solvable via a Bayesian Optimization procedure detailed within.  This dual approach has the added benefit of quantifying the Sim2Real Gap between a system simulator and its hardware counterpart.  Our contributions are twofold.  First, we show repeatability with respect to our outlined optimization procedure in solving these optimization problems.  Second, we show that the same procedure can discriminate between different environments by identifying the Sim2Real Gap between a simulator and its hardware counterpart operating in different environments. {\color{white}\cite{video}}

\end{abstract}

\section{INTRODUCTION}
It is a well-known problem that simulators are an imperfect representation of their real counterparts.  As a result, both the study of identifying simulator accuracy and of developing a controller in simulation such that it translates well to reality have been of increasing importance in the recent past~\cite{kadian2020sim2real, tobin2017domain, sadeghi2016cad2rl, andrychowicz2020learning, matas2018sim, nachum2019multi}.  This discrepancy is termed the \textit{Sim2Real} gap, and the process of developing controllers and/or policies in simulation such that they transfer well to reality is termed \textit{Sim2Real transfer}.  There has also been a wealth of work aimed at building better simulators to facilitate such transfer, as system evaluation and development within a simulator is significantly less expensive, time-intensive, and dangerous, especially for safety-critical systems~\cite{gupta2017cognitive,dosovitskiy2017carla,xia2018gibson,savva2017minos}.

This sequence of evaluation and development of a system's controller within a simulator underscores the current theoretical push for verifiable Test and Evaluation techniques.  More aptly, these techniques would determine whether these (and perhaps other) controllers adequately produce desired system behavior in reality~\cite{seshia2016towards}.  For context, the desired system behavior is oftentimes expressed as a temporal logic specification~\cite{baier2008principles,corso2020survey}.  The pursuit of such verification techniques has been studied from both a model-based perspective~\cite{annpureddy2011s,tuncali2016utilizing,donze2010breach,dreossi2019verifai} and from a purely data-driven perspective as well~\cite{ghosh2018verifying,gangopadhyay2019identification,deshmukh2017testing}.  Additionally, the authors note that the verification problem is oftentimes phrased as an optimization problem, and that a specific solution technique, Bayesian Optimization, has also been employed to solve the dual problem of verification, control-development~\cite{marco2017virtual,berkenkamp2016safe,berkenkamp2021bayesian, corso2020survey}.

However, as expressed in~\cite{corso2020survey}, direct application of the data-driven techniques proposed in~\cite{ghosh2018verifying,gangopadhyay2019identification,deshmukh2017testing} to verify real systems might require a prohibitively large number of samples to make any evaluation or verification claim.  Additionally, even though there exist techniques to offset this sample cost in high dimensions when using a Bayesian Optimization specific approach, \textit{e.g.} through random embeddings~\cite{wang2013bayesian, rana2017high, rolland2018high}, these techniques would make no use of the high-fidelity simulators in development~\cite{gupta2017cognitive,dosovitskiy2017carla,xia2018gibson,savva2017minos}.  As such, the authors had asked the question in prior work: could one use a system simulator to offset the required number of rollouts of a true system required for verification of its controller~\cite{akella2021learning}? Could we also determine simulator accuracy through this procedure? Motivated by these questions, this letter aims to develop a simulator-based evaluation procedure that also bounds the Sim2Real gap.

\begin{figure}
    \centering
    \includegraphics[width = 0.49\textwidth]{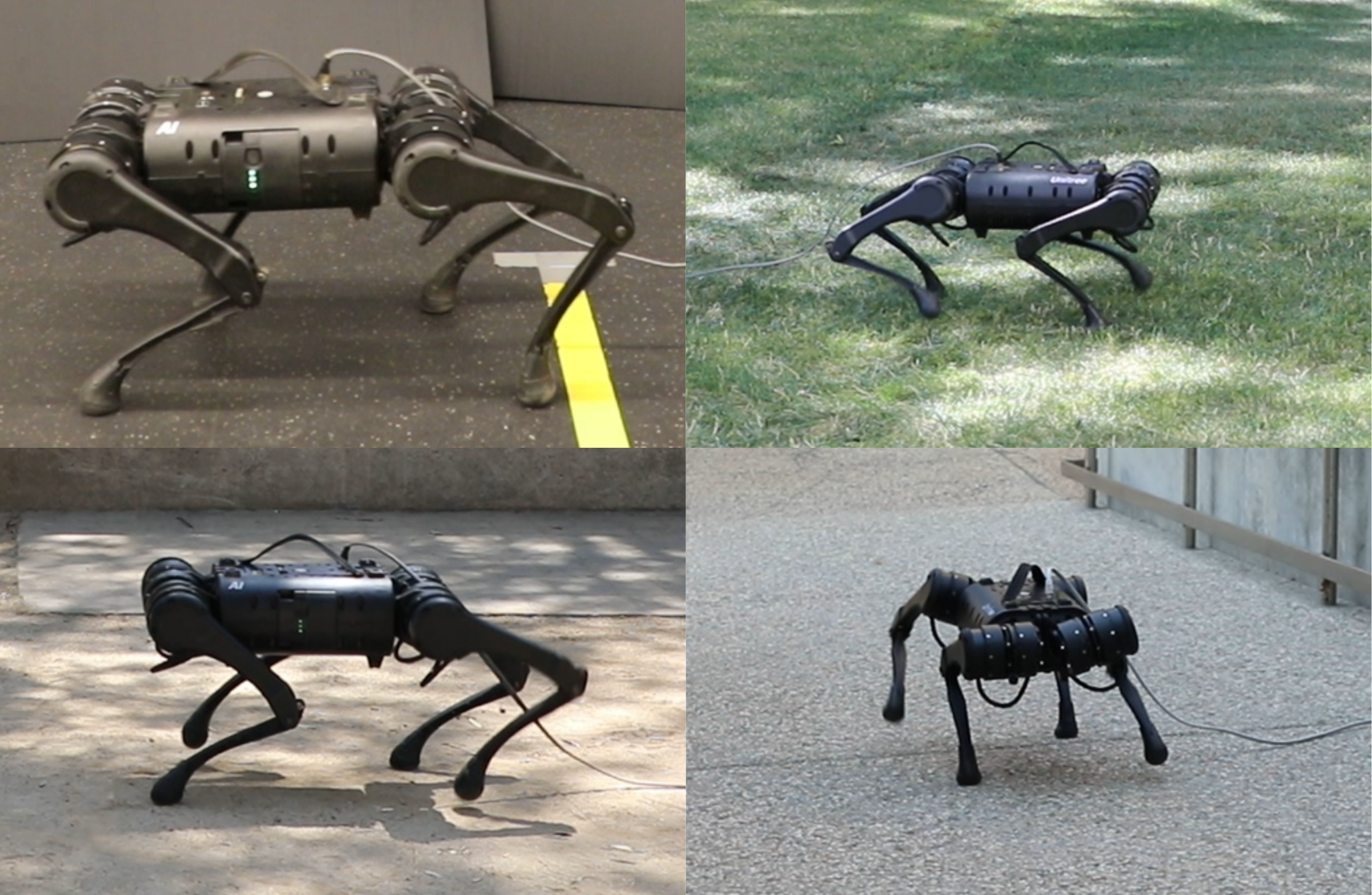}
    \vspace{-0.3 cm}
    \caption{Unitree A1 Quadruped shown walking in the different environments in which its walking gait is evaluated.}
    \label{fig:Quadrupedphoto}
    \vspace{-0.5 cm}
\end{figure}

\spacing
\newidea{Our Contribution:}
Our contribution is threefold.
\begin{itemize}
    \item First, we develop a Bayesian Optimization algorithm based on prior controls works.  We prove that this algorithm produces an upper bound to a maximization problem that is close to the true maximum, and that this bound holds with a minimum probability.
    \item Second, we develop an optimization procedure designed to lower bound the robustness with which the true system satisfies its objective while minimizing two simulator-based optimization problems.  We prove that this procedure identifies a lower bound that is close to the real system's minimum robustness in expectation.
    \item Third, we lower bound the walking robustness of a Quadruped by identifying its minimum simulator robustness and determining its Sim2Real gap in a variety of environments.  We show that our algorithm can repeatably identify these parameters as well.
\end{itemize}

\spacing
\newidea{Organization}
Section~\ref{sec:preliminaries} briefly outlines Bayesian Optimization.  Section~\ref{sec:definitions} states some definitions and assumptions facilitating a formal statement of our problem in Section~\ref{sec:prob_state}.  Section~\ref{sec:optimization} details our proposed optimization procedure and Section~\ref{sec:system_evaluation} details its use in evaluating system performance and bounding the Sim2Real Gap.  Finally, Section~\ref{sec:experimental} shows an example of our procedure evaluating the Unitree A1 Quadruped shown in Figure~\ref{fig:Quadrupedphoto}.
\section{BACKGROUND INFORMATION}
\label{sec:preliminaries}
In this section, we will briefly describe Bayesian Optimization - a necessary solution technique that lays the groundwork for future work mentioned in the paper.  To facilitate its description, we will start with some notation.

\spacing
\newidea{Notation:} $\mathbb{R}_+ = \{x \in \mathbb{R}~|~x \geq 0\}$ and $\mathbb{R}_{++} = \{x \in \mathbb{R}~|~x > 0\}$.  A signal $s:\mathbb{R}_{+} \to \mathbb{R}^n$.  The space of all signals $\signalspace = \{s~|~s:\mathbb{R}_{+} \to \mathbb{R}^n\}$. $C(X)$ is the set of all continuous functions over $X$.  A kernel function $k: \mathbb{Z} \times \mathbb{Z} \to \mathbb{R}_{+}$ is a positive semi-definite, symmetric function.  $\mathcal{R}(k)$ is the Reproducing Kernel Hilbert space (RKHS) of a kernel $k$ and $\|J\|_{RKHS}$ is the RKHS norm of a function $J:\mathbb{Z} \to \mathbb{R}$.

\spacing
\newidea{Bayesian Optimization:} The brief description of Bayesian Optimization (BO) in this subsection stems primarily from~\cite{srinivas2009gaussian, chowdhury2017kernelized}.  Bayesian Optimization attempts to solve optimization problems of the following form:
\begin{equation}
    J^* = \max_{z \in \mathbb{Z}}~J(z),~\mathbb{Z}\subset \mathbb{R}^l,~l<\infty.
\end{equation}
The optimization procedure follows a series of steps.  First, either a Gaussian Process is provided or fit to an initial data-set $\mathbb{D}_n = \{(z_i,y_i)\}_{i=1}^n$ with (potentially) noisy samples $y_i$ of the following form:
\begin{equation}
    \label{eq:sampling_criteria}
    y_i = J(z_i) + \xi_i \sim \mathcal{N}(0,\lambda\nu^2),~\lambda,\nu \in \mathbb{R}_{+}.
\end{equation}
Here, $\lambda,\nu$ are parameters for the Bayesian Optimization procedure - specifically for Gaussian Process Regression.  To this data-set $\mathbb{D}_n$, the procedure then fits a Gaussian Process $\pi$ to $J$ based on choice of a kernel function $k:\mathbb{Z} \times \mathbb{Z} \to \mathbb{R}_{+}$:
\begin{align}
    \label{eq:GP_mean}
    \mu_n(z) & = k_n(z)^T\left(K_n + \lambda I\right)^{-1}y_{1:n}, \\
    k_n(z,z') & = k(z,z') - k_n(z)^T\left(K_t + \lambda I\right)^{-1}k_n(z'), \\
    \label{eq:GP_sigma}
    \sigma_n(z) & = k_n(z,z).
\end{align}
Here, $k_n(z) = [k(z,z_1), \dots k(z,z_n)]^T$ is the covariance of $z$ with respect to the sampled data $z_i\in \mathbb{D}_n$, $y_{1:n} = [y_1, y_2, \dots, y_n]^T$ are the noisy samples, and $(K_n)_{i,j} = k(z_i,z_j),~z_i,z_j\in \mathbb{D}_n$ is the positive-definite Kernel Matrix.  Third, the next sample point $z_{i+1}$ is defined as the maximizer of an \textit{acquisition function} over the fitted Gaussian Process $\pi$ to the function $J$:
\begin{equation}
    \label{eq:UCB}
    z_{i+1} = \argmax_{z \in \mathbb{Z}}~\mu_i(z) + \beta_{i+1} \sigma_i(x).
\end{equation}
The Upper Confidence Bound (UCB) acquisition function is shown above and is one example of an acquisition function~\cite{srinivas2009gaussian,bull2011convergence}.  Finally, the procedure samples $z_{i+1}$, generates a new measurement $y_{i+1}$, adds it to the data-set, fits another Gaussian Process, and repeats the procedure.

Bayesian Optimization procedures guarantee eventual convergence by proving sub-linear growth in the sum-total regret $R_j = \sum_{i=1}^j r_i$ where $r_i = J^*-J(z_i)$.   As we assume we have noisy samples $y_i$ of our function $J$, these regret growth bounds are written with respect to the \textit{maximum information gain} at iteration $i$:
\begin{gather}
    \label{eq:max_info_gain}
    \gamma_i = \max_{A \subset \mathbb{Z} \suchthat |A| = i}~I(y_A;J_A).
\end{gather}
Here, $I(y_A;J_A)$ is the \textit{mutual information gain} between $J_A = [J(z)]_{z\in A}$ and $y_A = J_A + \xi_A \sim \normal(0,\lambda v^2I)$. $I(y_A;J_A)$ quantifies the reduction in uncertainty about the objective $J$ after sampling points $z \in A$.  With this brief description of Bayesian Optimization, we will move to formally stating the problem under study in this paper.

\section{PROBLEM FORMULATION}
We will split this section into two parts.  First, we will provide some definitions and assumptions that will be used throughout the paper.  Then we will state our problem.

\subsection{Definitions and Assumptions}
\label{sec:definitions}
As mentioned, the goal of Test and Evaluation is to determine whether a system's controller can realize desired system behavior despite a set of (perhaps) adversarial phenomena in the working environment~\cite{seshia2016towards}.  To formalize this notion, we will first define the environment state $x_E$.
\begin{definition}
\label{def:environment}
The \textit{environment} $E$ is the state of the world in which the system operates including the state of the system itself, \textit{e.g.} the cave in which a robot is traversing coupled with any motor failures the robot may have suffered, the airspace in which a jet flies along with any engine failures, \textit{etc}.  The state of the environment will be represented through the \textit{environment state vector} $x_E$.
\end{definition}
As defined, the environment state $x_E$ may be incomprehensibly large, indeed even infinite.  Additionally, the goal of Test and Evaluation is to verify whether the system under test can operate satisfactorily despite a set of allowable perturbations in its environment.  As such, we will assume we can partition the environment state into a set of knowable, \textit{i.e.} testable, and unknowable phenomena.
\begin{definition} 
\label{def:valid_tests}
The state $x_E$ of the environment $E$ can be segmented into a set of known disturbances $d \in \mathcal{D}$ and unknown disturbances $w \in \mathbb{W}$, \textit{i.e.} $x_E^T = [d^T, w^T]^T$.  The space of known disturbances $\mathcal{D}$ is the \textit{feasible test space} and each $d \in \mathcal{D}$ is a \textit{test parameter vector}.
\end{definition}

Now that we have formally defined our environment, it remains to classify the types of systems under study.  For the sequel, we will consider a general, (perhaps) nonlinear control system coupled with a controller.
\begin{equation}
\label{eq:true_system}
\dot x = f(x,u,d,w),~~
\begin{array}{cc}
    x \in \mathcal{X}, & u = U(x,d) \in \mathcal{U},  \\
    d \in \mathcal{D}, & w \in \mathbb{W}~\mathrm{and}~w \sim \pi(d).
\end{array}
\end{equation}
For the true system~\eqref{eq:true_system}, $\mathcal{U}:\mathcal{X} \times \mathcal{D} \to \mathcal{U}$ and is our true-system controller, $d \in \mathcal{D}$ is a specific test as in Definition~\ref{def:valid_tests}, and $w \in \mathbb{W}$ is our unknown disturbance as per Definition~\ref{def:valid_tests}.  We also assume $w$ to be a random variable distributed via the unknown, (perhaps) $d$-dependent distribution $\pi(d)$.  Furthermore, we note that $\mathcal{X} \subset \mathbb{R}^n$ and $\mathcal{U} \subset \mathbb{R}^m$.  We will likewise assume we have a simulator for this true system.
\begin{equation}
\label{eq:nominal_system}
\dot{\hat{x}} = \hat{f}\left(\hat{x},\hat{u},d,\hat{w}\right),
\begin{array}{cc}
    \hat{x} \in \mathcal{X}, & \hat{u} = \hat{U}\left(\hat{x},d\right) \in \mathcal{U},  \\
    d \in \mathcal{D}, & \hat{w} \in \mathbb{\hat{W}}~\mathrm{and}~\hat{w} \sim \hat{\pi}(d).
\end{array}
\end{equation}
As before, $\hat{U}:\mathcal{X} \times \mathcal{D} \to \mathcal{U}$ and is our simulator controller, $d \in \mathcal{D}$ is a test as per Definition~\ref{def:valid_tests}, and $\hat{w} \in \mathbb{\hat{W}}$ is our unknown simulator disturbances as per Definition~\ref{def:valid_tests}.  We likewise assume $\hat w$ is distributed via the unknown, (perhaps) $d$-dependent distribution~$\hat{\pi}(d)$.  Here we note that we have also implicitly defined a simulator environment that also satisfies Definition~\ref{def:environment} as we have unknown simulator disturbances that are different from those that exist in the real world.  For context, this setup models most systems with a Gazebo simulator, as Gazebo is non-deterministic.

In order to determine whether either system satisfies its specification, we require the system's signal trace, \textit{i.e.} its state trajectory $\phi$.  In defining $\phi$ we will abbreviate $U(x(t),d) = U(t)$ and $\hat U\left( \hat x(t), d\right) = \hat U (t)$.
\begin{equation}
\label{eq:trajectories}
\begin{aligned}
    \phi^U_t(x_0) & = x_0 + \int_{0}^t f(\phi^U_s(x_0), U(s), d, w(s))~ds, \\
    \hat{\phi}^{\hat{U}}_t\left(\hat{x}_0\right) & = \hat{x}_0 + \int_{0}^t \hat{f}\left(\hat{\phi}^{\hat{U}}_s\left(\hat{x}_0\right), \hat{U}(s), d, \hat{w}(s)\right)~ds.
\end{aligned}
\end{equation}
Here, $w(t) = w \in \mathbb{W}$ such that $w \sim \pi(d)$ as per equation~\eqref{eq:true_system}, and the same holds for $\hat{w}(t)$ with respect to equation~\eqref{eq:nominal_system} as well.  Furthermore, we note that while $\phi^U_t(x_0) \in \mathcal{X}$ and represents the system state after some elapsed time, when we drop the time suffix, $\phi^U(x_0) \in \signalspace$ and represents the state trajectory signal as a whole.  Owing to the noise sequences $w(t),\hat{w}(t)$ then, the resulting closed loop signals $\phi^U(x_0),\hat{\phi}^{\hat{U}}\left(\hat{x}_0\right)$ are random variables.  To formalize this notion we will state that the trajectories are distributed via unknown, (perhaps) $d$-dependent distributions $\Pi(d),\hat{\Pi}(d)$.
\begin{equation}
\label{eq:probability_spaces}
\begin{gathered}
        \phi^U(x_0),\hat{\phi}^{\hat{U}}\left(\hat{x}_0\right) \in \signalspace, \\
        \phi^U(x_0) \sim \Pi(d)~\mathrm{and}~\hat{\phi}^{\hat{U}}\left(\hat{x}_0\right) \sim \hat{\Pi}(d).
\end{gathered}
\end{equation}

Then, we will end with one assumption on our capacity to measure the satisfaction of a system's objective.  To formalize this notion, we will define a robustness measure.
\begin{definition}
\label{def:robustness}
A system specification and/or objective $\psi: \signalspace \to \{\true, \false\}$ has an associated \textit{robustness measure}~$\rho: \signalspace \to \mathbb{R}_+$ such that a signal $s$ satisfies $\psi$, \textit{i.e.} $\psi(s) = \true$ if and only if $\rho(s) \geq 0$, \textit{i.e.}
\begin{equation}
    \psi(s) = \true \iff \rho(s) \geq 0.
\end{equation}
\end{definition}
\noindent Here, the authors note that such a robustness measure could be that which exists for every Signal Temporal Logic specification~\cite{donze2010robust,baier2008principles}, it could be the minimum value of a control barrier function over some bounded time interval~\cite{ames2016control}, or it could be some other function mapping signals to a quantifiable satisfaction metric, \textit{i.e.} distance to a leading vehicle~\cite{corso2020survey,wheeler2019critical}.  The specific type of robustness measure does not matter, and to help clarify this setting, we will provide an example.

\begin{example}
\label{example1}
Consider a simple autonomous agent, say a turtlebot, navigating within a predefined space $\mathcal{X} = [-1,1]^2$.  Further assume the agent's goal is to avoid an obstacle $O = \{x \in [-1,1]^2~|~\|x - o\|_2 < \delta_o\}$ while navigating to a specific goal region $G = \{x \in [-1,1]^2~|~\|x-g\|_2 \leq \delta_g\}$ within $T$ seconds.  Also assume the obstacle's center location $o$ can vary over $\mathcal{D} = [-1,1]^2$.  Then the test parameter vector $d \in \mathcal{D}$ as per Definition~\ref{def:valid_tests} is the obstacle's center location.  The robustness measure $\rho(s) = \min_{t \in [0,T]} \{\delta_g - \|s(t)-g\|, \|s(t) - o\| - \delta_0\}$ which satisfies Definition~\ref{def:robustness}.
\end{example}

\noindent Our formal problem statement will follow.

\subsection{Problem Statement}
\label{sec:prob_state}
As motivated prior, our goal is to determine whether the true system satisfies its specification and/or objective $\psi$ as defined in Definition~\ref{def:robustness}.  As we expect our true system to be noisy however, we cannot directly optimize over closed loop trajectories, \textit{i.e.} $\rho\left(\phi^U(x_0)\right)$.  As a result, we will choose to optimize for the minimum expected value instead, \textit{i.e.},
\begin{equation}
    \label{eq:minimum_true_robustness}
    \rho^* = \min_{d \in \mathcal{D}}~\expect_{\Pi(d)}\left[\rho\left(\phi^U(x_0)\right) \right].
\end{equation}
\noindent This leads to our formal problem statement.
\begin{problem} 
For a robustness measure $\rho$ satisfying Definition~\ref{def:robustness} determine an appropriate lower bound $\rho^e$ to $\rho^*$ as defined in equation~\eqref{eq:minimum_true_robustness} and error constant $\epsilon$ such that:
\begin{equation}
    \rho^* \geq \rho^e,\quad |\rho^* - \rho^e| \leq \epsilon, \quad 0 \leq \epsilon < \infty.
\end{equation}
\end{problem}

\spacing
\newidea{Overarching Approach:}  In what will follow, we will provide a brief overview of our procedure.  First, the authors note that one could directly apply the Bayesian Optimization approach we will detail to solve for a lower bound for $\rho^*$ as in equation~\eqref{eq:minimum_true_robustness}.  This procedure requires some assumptions on optimization problem~\eqref{eq:minimum_true_robustness} though we will neglect to mention those at the moment.  However, direct application of this Bayesian technique might result in a prohibitively large number of required true-system runs to realize an effective lower bound~\cite{corso2020survey}.  As a result, we will opt instead to solve two simulator-specific optimization problems as follows:
\begin{equation}
    \label{eq:simulator_opt}
    \begin{gathered}
    \hat{\rho}^* = \min_{d \in \mathcal{D}}~\expect_{\hat{\Pi}(d)}\left[\rho\left(\hat{\phi}^{\hat{U}}\left(\hat{x}_0\right)\right) \right], \\
    e^* = \max_{d \in \mathcal{D}}~\expect_{\Pi(d),\hat{\Pi}(d)}\left[\left|\rho\left(\phi^U(x_0)\right) - \rho\left(\hat{\phi}^{\hat{U}}\left(\hat{x}_0\right)\right) \right| \right].
    \end{gathered}
\end{equation}
Making some assumptions on the optimization problems in equation~\eqref{eq:simulator_opt}, we will show we can identify $\hat{\rho}^e,e^e$ that satisfy the following sets of inequalities:
\begin{equation}
    \label{eq:bounds}
    \begin{gathered}
    \hat{\rho}^* \geq \hat{\rho}^e~\&~|\hat{\rho}^* - \hat{\rho}^e| \leq \epsilon_0 \withprob \geq 1-\delta_0, \\
    e^* \leq e^e~\&~|e^*-e^e| \leq \epsilon_1 \withprob \geq 1-\delta_1,
    \end{gathered}
\end{equation}
with $\delta_0,\delta_1 \in (0,1]$ and $\epsilon_0,\epsilon_1 \in \mathbb{R}_+$.  Then our final result stems by defining $\rho^e = \hat \rho^e - e^e$, $e = 2 e^e + \epsilon_0 + \epsilon_1$, and noting that this satisfies the required conditions in our problem statement, \textit{i.e.}
\begin{equation}
    \rho^* \geq \rho^e,~\mathrm{and}~|\rho^* - \rho^e| \leq \epsilon \withprob \geq (1-\delta_0)(1-\delta_1).
\end{equation}

\section{MAIN CONTRIBUTIONS}
In this section, we will detail our proposed Bayesian Optimization Algorithm and state and prove a theorem regarding its use.  Then, we will state and prove another theorem regarding its application to lower bounding $\rho^*$ in equation~\eqref{eq:minimum_true_robustness} while identifying the maximum Sim2Real gap $e^*$ in equation~\eqref{eq:simulator_opt}.  We will split this section into two sections.  The first will detail the optimization algorithm and the latter will detail its use in lower bounding true-system robustness.
\subsection{OPTIMIZATION ALGORITHM}
\label{sec:optimization}
In this section, we will detail our proposed GP-UCB Bayesian Optimization algorithm building off the work done in~\cite{srinivas2009gaussian,chowdhury2017kernelized} and algorithms utilized in prior controls works~\cite{ghosh2018verifying,berkenkamp2016safe,berkenkamp2021bayesian}.  More aptly, this algorithm will identify upper bounds $J^e$ to the following optimization problem:
\begin{equation}
    \label{eq:basic_optimization}
    J^* = \max_{z \in \mathbb{Z}}~J(z).
\end{equation}
We construct such an algorithm, for as motivated in Section~\ref{sec:prob_state}, we will require accurate estimates $\hat \rho^e,e^e$ as in equation~\eqref{eq:bounds} for our procedure.  Before stating the algorithm however, we will briefly describe it.  To start, we require positive constants $\delta \in (0,1]$, $B,R, \epsilon \in \mathbb{R}_{++}$ and an initial dataset $\mathbb{D}_0=\{(z,y)\}$ of one (perhaps noisy) sample of the objective $J$ as per equation~\eqref{eq:sampling_criteria}.  Then, Algorithm~\ref{alg:algorithm} first defines in Line 2 a scale factor
\begin{equation}
    \beta_i = B + R \sqrt{2 \ln{\frac{\sqrt{\det\left((1+\frac{2}{i})I + K_i\right)}}{\delta}}}, \label{eq:beta}
\end{equation}
and, in Line 3, identifies the maximizer of the UCB acquisition function $z_i$ with respect to this $\beta_i$ and the fitted Gaussian Process $\pi$ to $J$ at iteration $i$.  In Line 4, the algorithm collects a noisy measurement $y_i$ of $J(z_i)$, and the sample pair $(z_i,y_i)$ is added to the data-set generating $\mathbb{D}_i$.  Line 5 defines the simple regret bound
\begin{equation}
    \label{eq:F}
    F_i = 2 \beta_i \sigma_{i-1}(z_i).
\end{equation}
Here, $\sigma_{i-1}$ is the variance of the fitted Gaussian Process to the data-set $\mathbb{D}_{i-1}$. Lines 6-9 check whether $F_i \leq \epsilon$, the desired tolerance, and if so, the algorithm outputs $\epsilon = \mu_{i-1}(z_i) + \beta_i \sigma_{i-1}(z_i)$ and terminates.  Otherwise, in Line 10, the algorithm updates the fitted Gaussian Process $\pi$ with respect to $\mathbb{D}_i$.  Finally, before moving to this section's main results, we will state an assumption underlying use of this Algorithm.  Indeed, this is a common assumption whenever utilizing Bayesian Optimization~\cite{ghosh2018verifying,gangopadhyay2019identification,berkenkamp2016safe,berkenkamp2021bayesian,marco2017virtual}.
\begin{assumption}
\label{assump:bayes}
For the optimization problem~\eqref{eq:basic_optimization}, the decision space $\mathbb{Z}$ is compact and convex.  Additionally, for some $B,R\in\mathbb{R}_{++}$ and kernel $k$, the objective function $J$ has $\|J\|_{RKHS} \leq B$, and the samples $y_i$ of $J(z_i)$, as per equation~\eqref{eq:sampling_criteria}, are corrupted by $R$-sub Gaussian Noise $\forall~i$.
\end{assumption}

With this assumption we can state the first key result of this paper.  Specifically, that Algorithm~\ref{alg:algorithm} will identify a $J^e$ such that $J^e \geq J^*$ and $|J^e - J^*| \leq \epsilon$ with probability $\geq 1-\delta$.
\begin{theorem}
\label{thm:algorithm}
Let Assumption~\ref{assump:bayes} hold, let $\delta \in (0,1]$, and let $\epsilon \in \mathbb{R}_{++}$.  At termination $i^*$, Algorithm~\ref{alg:algorithm} outputs $J^e$ such that $\prob_{\pi}[J^* \leq J^e] \geq 1-\delta$ and $\prob_{\pi}[|J^*-J^e| \leq \epsilon] \geq 1-\delta$, with $J^*$ as in equation~\eqref{eq:basic_optimization}, and $\pi$ as in Line 10.
\end{theorem}

Proving Theorem~\ref{thm:algorithm} requires two Propositions.  The first bounds the variance of the objective function $J$ with respect to the fitted Gaussian Process $\pi$ (Line 10) and the scale factor $\beta_i$ defined in equation~\eqref{eq:beta}.
\begin{proposition}[Theorem~2 in~\cite{chowdhury2017kernelized}]
\label{prop:bounding}
Let $\beta_i$ be as in~\eqref{eq:beta}, $\gamma_j$ as in equation~\eqref{eq:max_info_gain}, $\delta \in (0,1]$, and let Assumption~\ref{assump:bayes} hold.  With probability $\geq 1-\delta$, $|\mu_{i-1}(z) - J(z)| \leq \beta_i\sigma_{i-1}(z)~\forall~i$ and $\forall~z \in \mathbb{Z}$, and
\begin{equation}
    \beta_i \leq B + R \sqrt{2 \left(\gamma_j + 1 + \ln{\frac{1}{\delta} }\right)},~\forall~ i = 1,2,\dots,j.
\end{equation}
\end{proposition}
\noindent In Proposition~\ref{prop:bounding}, $\mu_{i-1}, \sigma_{i-1}$ are the fitted mean and variance functions for the Gaussian Process $\pi$ estimating the objective function $J$ based on the dataset $\mathbb{D}_{i-1}$.  For context, both inequalities in Proposition~\ref{prop:bounding} were taken from the proof for Theorem 2 in~\cite{chowdhury2017kernelized}.  The second proposition bounds the growth rate of $\gamma_j$ as defined in equation~\eqref{eq:max_info_gain}.
 
\begin{algorithm}[t]
\caption{Modified GP-UCB Bayesian Optimization}\label{alg:algorithm}
\begin{algorithmic}[1]
\Require $\delta \in (0,1]$, $B,R \in \mathbb{R}_{++}$, an initial data-set $\mathbb{D}_0 = \{(z,y)~|~z \in \mathbb{Z},~y$ as per~\eqref{eq:sampling_criteria}$\}$, and tolerance $\epsilon \in \mathbb{R}_{++}$.

\hspace{-0.56 in} \noindent \textbf{Returns:} A Gaussian Process $\pi$, and an upper bound $J^e$ such that $\prob_{\pi}[J^* \leq J^e] \geq (1-\delta)$.
\Initialize $i=1$, $\eta_i = \frac{2}{i}$, Gaussian Process with mean $\mu_0$ and covariance $\sigma_0$ from the data-set, $\mathbb{D}_0$ as per equations~\eqref{eq:GP_mean} and \eqref{eq:GP_sigma}.
\While{True}
\State $\beta_i \gets B + R \sqrt{2 \ln{\frac{\sqrt{\det\left((1+\eta_i)I + K_i\right)}}{\delta}}}$
\State $z_i \gets \argmax_{z \in \mathbb{Z}}~\mu_{i-1}(z) + \beta_i\sigma_{i-1}(z)$
\State $\mathbb{D}_{i} \gets \mathbb{D}_{i-1} \cup (z_i, y_i$ as per equation~\eqref{eq:sampling_criteria})
\State $F_i \gets 2 \beta_i \sigma_{i-1}(z_i)$
\If{$F_i \leq \epsilon$}
    \State $J^e = \mu_{i-1}(z_i) + \beta_i \sigma_{i-1}(z_i)$
    \State \textbf{return} $\epsilon$ 
\EndIf
\State Update the Gaussian Process $\pi$ with mean $\mu_i$ and variance $\sigma_i$ as per~\eqref{eq:GP_mean} and \eqref{eq:GP_sigma} with respect to $\mathbb{D}_i$
\State $i \gets i+1$
\EndWhile
\end{algorithmic}
\end{algorithm}

\begin{proposition}[Theorem~5 in~\cite{srinivas2009gaussian}]
\label{prop:bound_information}
Let Assumption~\ref{assump:bayes} hold. There exists a kernel $k$ such that the growth in the maximum information gain $\gamma_j$ satisfies the following inequality:
\begin{equation}
    \gamma_j \leq O(j^p \log(j)),~p < 0.5.
\end{equation}
\end{proposition}
As before, Proposition~\ref{prop:bound_information} stems directly from Theorem 5 in~\cite{srinivas2009gaussian} which provides the growth bound for the information gain $\gamma_j$ for common kernels.  With these propositions, we can now state and prove two Lemmas required for proving Theorem~\ref{thm:algorithm}.  The first Lemma will bound the simple regret $r_i$ by our simple regret bound $F_i$ defined in equation~\eqref{eq:F}.
\begin{lemma}
\label{lem:bound_simple_regret}
Let Assumption~\ref{assump:bayes} hold, and let $F_i$ be as in~\eqref{eq:F}. The simple regret $r_i$ satisfies the following inequality with respect to the Gaussian Process $\pi$ (Line 10):
\begin{equation}
    \prob_{\pi}[r_i \leq F_i] \geq 1-\delta.
\end{equation}
\end{lemma}
\begin{proof}
By definition of the simple regret $r_i$, the optimal sample $z_i$ (Line 3), the simple regret bound $F_i$, and the first inequality in Proposition~\ref{prop:bounding}, we have the following:
\begin{align}
    r_i & = J^* - J(z_i), \\
    & \leq \beta_i \sigma_{i-1}(z_i) + \mu_{i-1}(z_i) - J(z_i), \withprob \geq 1-\delta \\
    & \leq 2 \beta_i \sigma_{i-1}(z_i) = F_i, \withprob \geq 1-\delta.
\end{align}
\end{proof}

\noindent We can also bound the growth of $F_i$ .
\begin{lemma}
\label{lem:Fbound}
Let Assumption~\ref{assump:bayes} hold and let $\delta \in (0,1]$. Then, 
\begin{equation}
    \sum_{i=1}^j F_i \leq O\left(\sqrt{j}\left(B\sqrt{\gamma_j} + R \sqrt{\gamma_j\left(\gamma_j + \ln{\frac{1}{\delta}}\right)}\right)\right),
\end{equation}
with probability $\geq 1-\delta$ with respect to the Gaussian Process $\pi$ (Line 10),  and with $F_i$ as in~\eqref{eq:F}.
\end{lemma}
\begin{proof}
From the definition of the simple regret bound $F_i$ and the second inequality in Proposition~\ref{prop:bounding}, we have that
\begin{equation}
    \sum_{i=1}^j F_i \leq 2\left(B + R \sqrt{2 \left(\gamma_j + 1 + \ln{\frac{1}{\delta} }\right)}\right) \sum_{i=1}^j \sigma_{i-1}(z_i),
\end{equation}
with probability $\geq 1-\delta$.
The result then stems via Lemma 4 in~\cite{chowdhury2017arxiv}, which states that $\sum_{i=1}^j \sigma_{i-1}(z_i) \leq O\left(\sqrt{j\gamma_j}\right)$.
\end{proof}
\noindent Now we can prove Theorem~\ref{thm:algorithm}.

\begin{proof}
The proof for this theorem requires two parts.  First, we need to prove that the Algorithm terminates, and second, we need to prove that the Algorithm outputs a $J^e$ satisfying the stated inequalities.  The first part of this proof follows a contradiction.  Specifically, assume $\nexists~i^* < \infty$ such that $F_{i^*} \leq \epsilon$ where $\epsilon > 0$ as per the assumptions in Theorem~\ref{thm:algorithm}.  In other words, this implies that $F_i > \epsilon > 0~\forall~ i = 1,2,\dots$.  Then consider the running average of $F_i$ and Lemma~\ref{lem:Fbound}:
\begin{align}
    \epsilon & < \lim_{j \to \infty}~\frac{1}{j}\sum_{i=1}^j F_i, \\
    & \leq \lim_{j \to \infty} O\left(\frac{\sqrt{j}\left(B\sqrt{\gamma_j} + R\sqrt{\gamma_j\left(\gamma_j + \ln{\frac{1}{\delta}}\right)}\right)}{j}\right).
\end{align}
Now, pick a kernel that satisfies the inequality in Proposition~\ref{prop:bound_information}, which is guaranteed to exist.  Then,
\begin{equation}
    \epsilon < \lim_{j \to \infty} O\left(\frac{j^z \log(j)}{j}\right) = 0,~\mathrm{as}~z < 1,
\end{equation}
which is a contradiction, as $\epsilon \in \mathbb{R}_{++}$.  This proves termination at some $i^* < \infty$.  It remains to identify an upper bound $J^e$ that satisfies the required inequality in Theorem~\ref{thm:algorithm}.

For the second part of the proof, due to the first inequality in Proposition~\ref{prop:bounding} and Line 7 in Algorithm~\ref{alg:algorithm}, we have the following inequality at termination $i^*$ and with probability $\geq 1-\delta$:
\begin{equation}
    J^* \leq \mu_{i^*-1}(z_{i^*}) + \beta_{i^*}\sigma_{i^*-1}(z_{i^*}) = J^e.
\end{equation}
This resolves one of the inequalities in Theorem~\ref{thm:algorithm}.  For the second inequality, by Proposition~\ref{prop:bounding} we know that at the sample point at termination $z_{i^*}$,
\begin{equation}
    |\mu_{i^*-1}(z_{i^*}) - J(z_{i^*})| \leq \beta_{i^*}\sigma_{i-1}(z_{i^*}) \withprob \geq 1-\delta.
\end{equation}
Taking one of the inequalities from above and adding $\beta_{i^*}\sigma_{i^*-1}(z_{i^*})$ yields the following, as $J(z_{i^*}) \leq J^*$:
\begin{equation}
    J^e - J^* \leq \epsilon \withprob \geq 1-\delta.
\end{equation}
Now, by definition of simple regret and Proposition~\ref{prop:bounding},
\begin{align}
    \epsilon \geq r_{i^*} & = J^* - J(z_{i^*}), \\
    & \geq J^* - J^e \withprob \geq 1-\delta.
\end{align}
As a result,
\begin{equation}
    |J^* - J^e| \leq \epsilon \withprob \geq 1-\delta,
\end{equation}
completing the proof.  Also to note, the probabilities throughout the proof are taken with respect to the Gaussian Process $\pi$ estimating the objective $J$ based on the dataset $\mathbb{D}_{i^*}$ at termination $i^*$.
\end{proof}

\noindent Now we will move to use Algorithm~\ref{alg:algorithm} to identify a close lower bound to $\rho^*$ as required of our overarching problem.

\subsection{EVALUATING SYSTEM PERFORMANCE}
\label{sec:system_evaluation}
Simply put, the approach we will take in this subsection will amount to two uses of Algorithm~\ref{alg:algorithm} where we identify a lower bound $\hat \rho^e$ to $\hat \rho^*$ and an upper bound $e^e$ to $e^*$.  Both $\hat \rho^*$ and $e^*$ are defined in equation~\eqref{eq:simulator_opt}.  In order to use Theorem~\ref{thm:algorithm} however, we need to state or prove that our optimization problems in equation~\eqref{eq:simulator_opt} satisfy Assumption~\ref{assump:bayes}.
\begin{assumption}
\label{assump:testing_assumption}
Both optimization problems in equation~\eqref{eq:simulator_opt} have a $B,R \in \mathbb{R}_{++}$ and a sampling scheme generating (noisy) measurements $y_i$ for each chosen sample $d_i$ such that they satisfy Assumption~\ref{assump:bayes}.
\end{assumption}
While this assumption seems restrictive, it has two separate rationalizations.  First, any universal kernel's Reproducing Kernel Hilbert Space is equivalent to the space of all continuous functions over the kernel's domain, \textit{i.e.} for a kernel $k: \mathbb{Z} \times \mathbb{Z} \to \mathbb{R}_+$, $\mathcal{R}(k) = C(\mathbb{Z})$~\cite{micchelli2006universal}. Effectively then, if we assume our objective functions for optimization problems~\eqref{eq:simulator_opt} are continuous in $d \in \mathcal{D}$, and we use a universal kernel $k$, then we know that our objective functions $J \in \mathcal{R}(k)$ and that $\exists~B \suchthat \|J\|_{RKHS} \leq B$.  Furthermore, the assumption that our objective functions are continuous in $d$ is not too restrictive.  Consider Example~\ref{example1} for instance where this holds.  The second rationalization stems from the fact that we are optimizing for the expected value of a random variable whose variance is bounded - $\rho$ cannot take infinite values as per Definition~\ref{def:robustness}.  As a result, Hoeffding's Lemma guarantees that any single realization of this random variable corresponds to a sub-gaussian random variable after mean-shifting~\cite{massart2007concentration}.  Therefore, there exists a variance proxy $R\in\mathbb{R}_{++}$ to satisfy the second half of Assumption~\ref{assump:testing_assumption}.  For context, this is the reason that we do not optimize for an arbitrary risk measure, as samples of this measure need not be sub-gaussian.  This would frustrate application of our developed approach.  Optimizing for such a measure, however, is the subject of future work.

With this assumption, we can state two Lemmas that will be used to prove our second contribution.
\begin{lemma}
\label{lemma:sim_robustness}
Let Assumption~\ref{assump:testing_assumption} hold, let $\delta_0 \in (0,1]$, and let $\epsilon_0 \in \mathbb{R}_{++}$.  Applying Algorithm~\ref{alg:algorithm} to solve the first optimization problem in equation~\eqref{eq:simulator_opt} yields $\hat \rho^e$ such that
\begin{equation}
    \prob_{\pi_0}\left[\hat{\rho}^* \geq \hat{\rho}^e\right]~\mathrm{and}~\prob_{\pi_0}\left[|\hat{\rho}^* - \hat{\rho}^e| \leq \epsilon_0 \right]~\mathrm{are}~\geq 1-\delta_0.
\end{equation}
where $\pi_0$ is the Gaussian Process generated by Algorithm~\ref{alg:algorithm}.
\end{lemma}
\begin{proof}
Use Theorem~\ref{thm:algorithm} with $\delta = \delta_0$ and $\epsilon = \epsilon_0$.
\end{proof}
\begin{lemma}
\label{lemma:sim_error}
Let Assumption~\ref{assump:testing_assumption} hold, let $\delta_1 \in (0,1]$, and let $\epsilon_1 \in \mathbb{R}_{++}$.  Applying Algorithm~\ref{alg:algorithm} to solve the second optimization problem in equation~\eqref{eq:simulator_opt} yields $e^e$ such that
\begin{equation}
     \prob_{\pi_1}\left[e^* \leq e^e\right]~\mathrm{and}~\prob_{\pi_1}\left[|e^* - e^e| \leq \epsilon_1 \right]~\mathrm{are}~ \geq 1-\delta_1,
\end{equation}
where $\pi_1$ is the Gaussian Process generated by Algorithm~\ref{alg:algorithm}.
\end{lemma}
\begin{proof}
Use Theorem~\ref{thm:algorithm} with $\delta = \delta_1$ and $\epsilon = \epsilon_1$.
\end{proof}

Now we can state and prove our main result.

\begin{theorem}
\label{thm:lower_bound}
Let $\hat \rho^e, \epsilon_0, \delta_0, \pi_0$ be as defined in Lemma~\ref{lemma:sim_robustness} and let $e^e, \epsilon_1, \delta_1, \pi_1$ be as defined in Lemma~\ref{lemma:sim_error}.  Define $\rho^e = \hat \rho^e - e^e$ and $\epsilon = 2e^e + \epsilon_0 + \epsilon_1$, then
\begin{gather}
    \prob_{\pi_0,\pi_1}\left[\rho^* \geq \rho^e \right] \geq (1-\delta_0)(1-\delta_1), \\
     \prob_{\pi_0,\pi_1}\left[|\rho^* - \rho^e| \leq \epsilon \right] \geq (1-\delta_0)(1-\delta_1).
\end{gather}
with $\rho^*$ as in equation~\eqref{eq:minimum_true_robustness}.
\end{theorem}
\begin{proof}
To start this proof, we note we can modify optimization problem~\eqref{eq:minimum_true_robustness} using the following result:
\begin{align}
    \min_{z \in Z}~J(z)\geq \min_{z \in Z}~J_2(z) - \max_{z \in Z}~|J(z) - J_2(z)|.
\end{align}
Specifically, replace $z \in Z$ with $d \in \mathcal{D}$, $J(z)$ with the objective function for optimization problem~\eqref{eq:minimum_true_robustness}, and $J_2(z)$ as the objective function for the first optimization problem in equation~\eqref{eq:simulator_opt}.  Then, via linearity of the expectation operator and independence of $\Pi(d),\hat\Pi(d)$, we get the following:
\begin{equation}
    \rho^* \geq \hat \rho^* - e^*,
\end{equation}
with $\hat\rho^*,e^*$ as defined in equation~\eqref{eq:simulator_opt}.  Then by Lemmas~\ref{lemma:sim_robustness} and~\ref{lemma:sim_error} we get our first inequality:
\begin{equation}
\prob_{\pi_0,\pi_1}\left[\rho^* \geq \rho^e \right] \geq (1-\delta_0)(1-\delta_1).
\end{equation}

For the second inequality, we know that,
\begin{equation}
    \label{eq:base_inequality}
    |\rho^* - \hat \rho^*| \leq e^*,
\end{equation}
where $\rho^*$ is defined in equation~\eqref{eq:minimum_true_robustness} and $\hat \rho^*, e^*$ are defined in equation~\eqref{eq:simulator_opt}.  For context, the inequality in~\eqref{eq:base_inequality} can be proven through a contradiction, though it offers little insight so it will not be produced here.  However, if you assume that $\rho^* - \hat \rho^* > e^*$, then evaluating each objective function at one of their minimizers results in a contradiction by definition of $e^*$.  Doing the same for the reverse inequality proves the statement.  Then by Lemma~\ref{lemma:sim_robustness},
\begin{align}
    |\rho^* - \hat \rho^*| \leq e^e + e_1 \withprob \geq 1-\delta_1.
\end{align}
Furthermore,
\begin{align}
    e^e + e_1 & \geq |\rho^* - \hat \rho^* + \hat \rho^e - \hat \rho^e| \withprob \geq 1-\delta_1\\
    & \geq |\rho^* - \hat \rho^e| - \epsilon_0 \withprob \geq (1-\delta_0)(1-\delta_1).
\end{align}
Then with probability $\geq (1-\delta_0)(1-\delta_1)$,
\begin{align}
    e^e+\epsilon_0+\epsilon_1 & \geq |\rho^* - \hat \rho^e + e^e - e^e|, \\
    & \geq |\rho^* - \rho^e| - e^e.
\end{align}
The proof concludes by defining $\epsilon = 2e^e + \epsilon_0 + \epsilon_1$.
\end{proof}

\begin{figure*}[t]
    \centering
    \includegraphics[width = \textwidth]{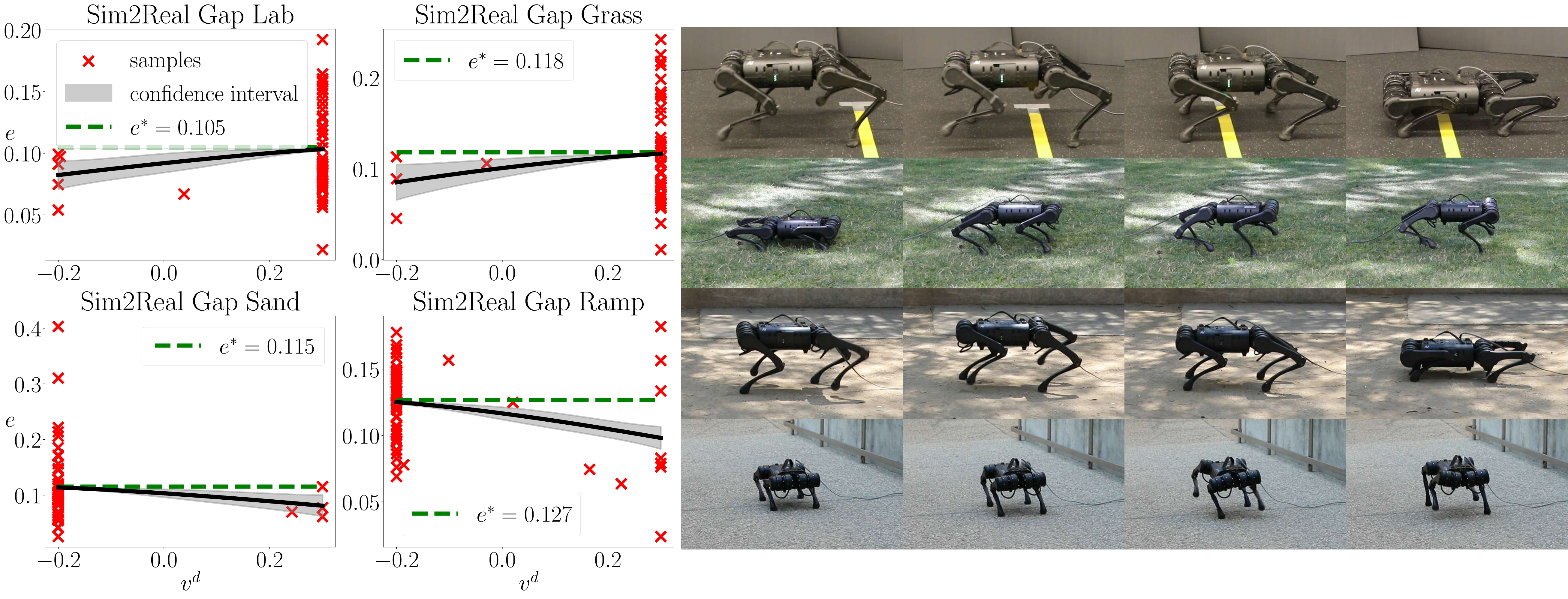}
    \vspace{-0.5 cm}
    \caption{The results and corresponding gait tiles for the quadruped walking in the four different environments in which we attempt to bound the Sim2Real gap with respect to the robustness measure $\rho$ in equation~\eqref{eq:example_optimization_robustness} in Section~\ref{sec:experimental}.  As expected, the quadruped performs best in the lab setting (lowest $e^*$) and worst on an $\sim4^o$ ramp where it has to fight against gravity while walking backwards (highest $e^*$).  Interestingly, lightly sanded ground and slightly muddy grass offer similar slipping conditions - enough so that the resulting differences are within the same tolerance $(0.003)$ of each other.}
    \label{fig:experiment_compilation}
    \vspace{-0.3 cm}
\end{figure*}
\begin{figure}[t]
    \centering
    \hspace*{-0.2 cm}\includegraphics[width = 0.495\textwidth]{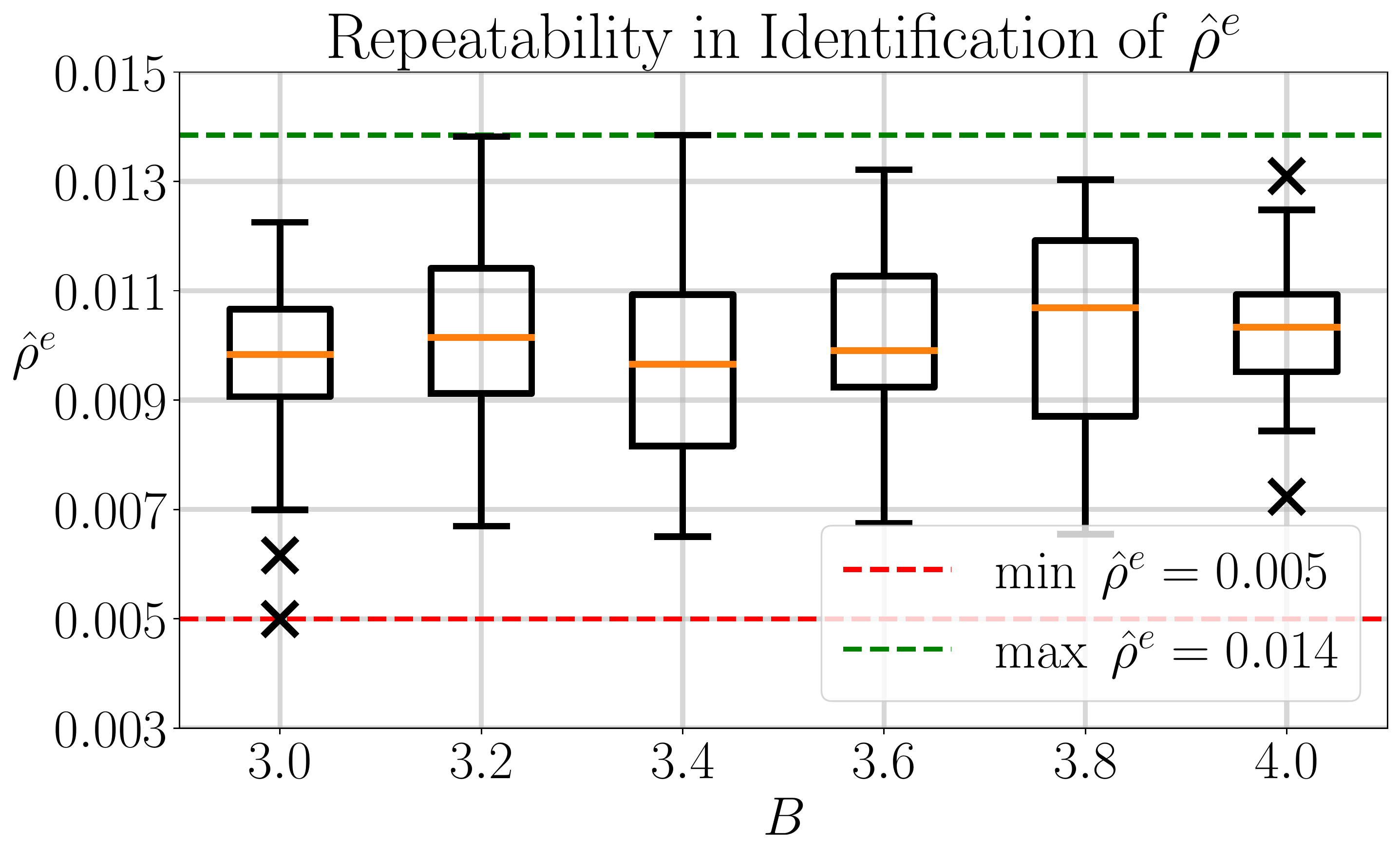}
    
    \vspace{-0.4 cm}
    \caption{Repeatability in our optimization analysis with respect to identification of $\hat \rho ^e$~\eqref{eq:example_optimization_robustness} in Section~\ref{sec:experimental}.  Over $120$ runs, $20$ for each choice of Reproducing Kernel Hilbert Space norm upper bound $B$, the procedure identifies an estimated, minimum simulator-robustness $\hat \rho^e$ that are all within tolerance $\epsilon_0 = 0.02$ of each other.  This is to be expected as per Theorem~\ref{thm:algorithm}.}
    \label{fig:repeatability}
    \vspace{-0.5 cm}
\end{figure}

This concludes the statement of our main results.  We will now illustrate Theorems~\ref{thm:algorithm} and~\ref{thm:lower_bound} through an example involving the Unitree A1 Quadruped depicted in Figure~\ref{fig:Quadrupedphoto}.

\section{EXPERIMENTAL RESULTS}
\label{sec:experimental}
For our experiment, we aim to test whether the Unitree Quadruped's forward velocity signal meets overshoot, settling time, and steady-state error criteria when driven by an IDQP-based trotting controller built off~\cite{buchli2009inverse} within the motion primitive framework in~\cite{ubellacker2021verifying}.  Mathematically, we will assume the ability to measure the state trajectory of our quadruped. We will call the true-system trajectory $\phi^U(x_0)$ and simulated state trajectory $\hat \phi^{\hat{U}}\left(\hat{x}_0\right)$. For any state signal $s$, we denote its forward-velocity component as $s_{v_x}$.   As such, our robustness measure per Definition~\ref{def:robustness} is as follows, with $v^d$ indicating the desired velocity, $\delta_o$ the maximum overshoot, and $\delta_s$ the allowable tolerance upon settling:
\begin{equation}
    \rho(s) = \min\left\{\begin{array}{c}
    \min\limits_{t \in [0,0.5]}~v^d + \delta_o - s_{v_x}(t), \\
    \min\limits_{t \in [0.5,1.5]}~\delta_s - \|v^d - s_{v_x}(t)\|.
    \end{array}\right\}
\end{equation}
\noindent Effectively, our robustness measure $\rho$ is positive for signals $s$ when the associated velocity does not exceed the desired velocity $v^d$ by more than $\delta_o$ in the first $0.5$ seconds of commanding the desired velocity and stays within a norm bound $\delta_s$ of the desired velocity $v^d$ for the next $1$ second.  In the event the desired velocity is negative, overshoot is calculated in the opposite direction (not shown).  Then, our optimization problems akin to~\eqref{eq:minimum_true_robustness}-\eqref{eq:simulator_opt} are as follows, where we abbreviate $\Delta \rho(\phi,\hat\phi) = \rho\left(\phi^U(x_0)\right) - \rho\left(\hat{\phi}^{\hat{U}}\left(\hat{x}_0\right)\right)$:
\begin{equation}
    \label{eq:example_optimization_robustness}
    \begin{gathered}
        \rho^* = \min_{v^d \in [-0.2,0.3]}\rho\left(\phi^U(x_0)\right), \\
         \hat{\rho}^* = \min_{v^d \in [-0.2,0.3]}\expect_{\hat{\Pi}(d)}\left[\rho\left(\hat{\phi}^{\hat{U}}\left(\hat{x}_0\right)\right) \right], \\
        e^* = \max_{v^d \in [-0.2,0.3]}\expect_{\Pi(d),\hat{\Pi}(d)}\left[\left|\Delta\rho(\phi,\hat\phi) \right| \right].
    \end{gathered}
\end{equation}
We aim to determine this lower bound in four different environments - (1) the AMBER lab at Caltech, (2) grass outside our building, (3) a stone ramp up to our building, and (4) a patch of sand nearby.

\begin{figure}[t]
    \centering
    \includegraphics[width = 0.49\textwidth]{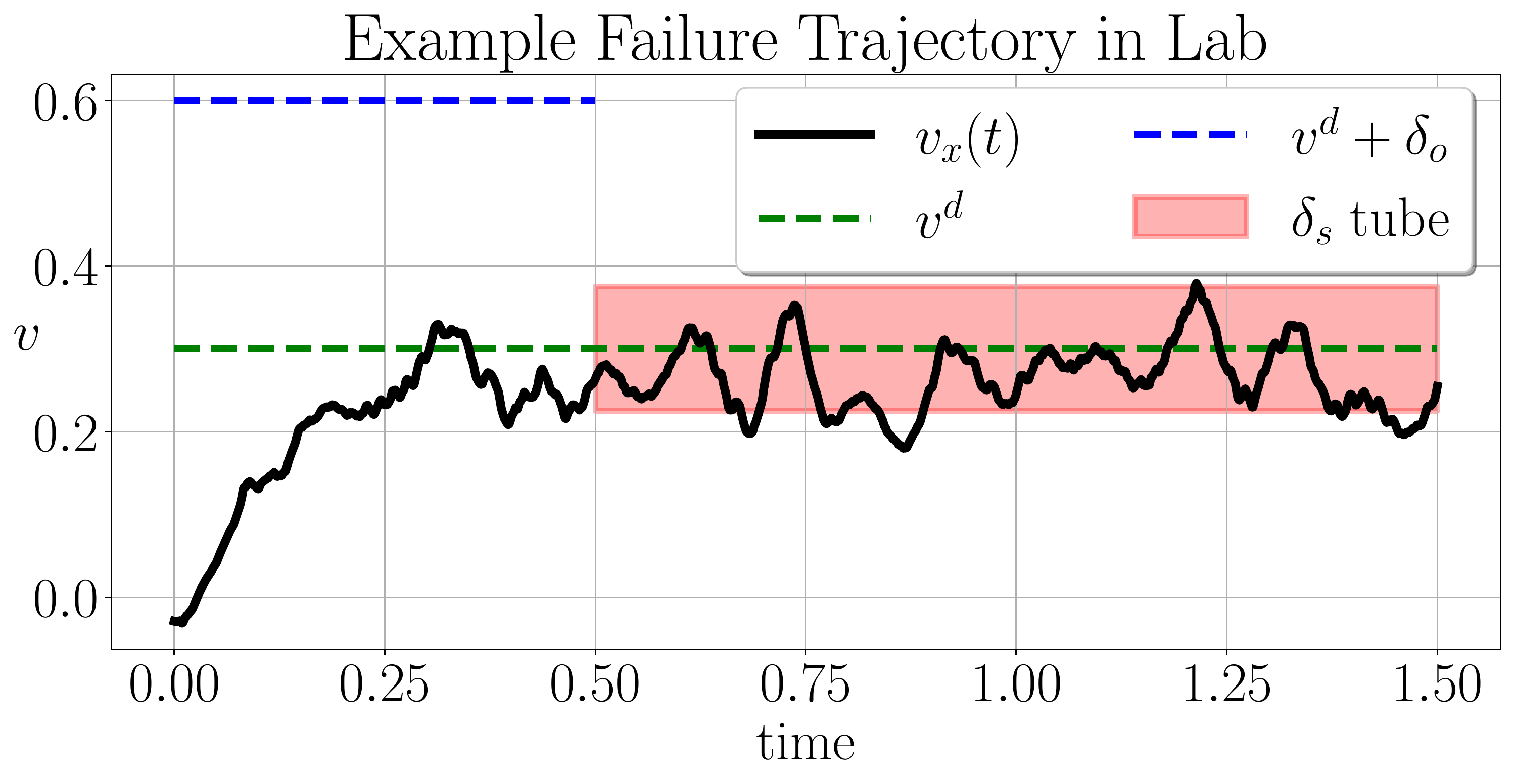}
    \vspace{-0.6 cm}
    \caption{Shown above is the velocity signal for the quadruped lab test depicted in Figure~\ref{fig:experiment_compilation}.  The quadruped's velocity signal meets the settling time and minimum overshoot criteria, but fails to meet the steady-state error criteria in its worst case test - maximum speed forward walking.}
    \label{fig:example_failure}
    \vspace{-0.6 cm}
\end{figure}

\spacing
\newidea{Repeatability of Results:} Figure~\ref{fig:repeatability} shows the results of running the optimization algorithm detailed in Section~\ref{sec:optimization} $120$ times to lower bound the minimum simulator robustness $\hat \rho^*$~\eqref{eq:example_optimization_robustness}.   Specifically,  Figure~\ref{fig:repeatability} shows the bounds generated for $20$ runs each for $6$ different Reproducing Kernel Hilbert space norm upper bounds $B$ with $R = 0.05, \delta = 1e^{-6}, \epsilon = 0.02$ for each run. Over all $120$ runs, the maximum spread of generated $\hat \rho^e$ lie within $0.009$ units of each other, indicating strong repeatability in the identification of this simulator's robustness parameter.  This is expected as the tolerance inequality in Theorem~\ref{thm:algorithm} dictates that all reported $\hat \rho^e$ values should be within $\epsilon$ of each other, and indeed they are.  This repeatability serves as our indicator regarding the efficacy of Algorithm~\ref{alg:algorithm} in solving the optimization problems we pose.

\spacing
\newidea{Sim2Real Gap:} We identified the Sim2Real Gap for the quadruped and its simulator in four environments.  Figure~\ref{fig:experiment_compilation} shows the results of upper bounding the Sim2Real gap, $e^*$ as per equation~\eqref{eq:example_optimization_robustness}, for (Top Left) the AMBER lab, (Top Right) the grass outside our building, (Bottom Left) a sandy patch outside our building, (Bottom Right) a ramp outside our building.  For each run, we initialized Algorithm~\ref{alg:algorithm} with $B = 1.5, R = 0.1, \delta = 1e^{-6}, \epsilon = 0.003$.  As expected, our Algorithm determined that the minimum Sim2Real gap occurs in the idealized lab setting - smallest $e^* = 0.105$ - and the largest error occurs when the robot walks backwards down a steep ramp - largest $e^* = 0.127$.  In all environments however, as per Theorem~\ref{thm:lower_bound}, we expect the true system to fail its specification.  The worst-case lab hardware trajectory depicted in Figure~\ref{fig:example_failure} confirms this notion as it yielded a robustness measure $\rho^{w} = -0.0454$ which we expect, as the minimum, hardware robustness in expectation $\rho^e \geq \hat \rho^e - e^* = 0.014 - 0.105 = -0.091$ as per Theorem~\ref{thm:lower_bound}.  The authors note that the likely reason for this failure is battery voltage degradation as it loses charge.  This causes the motors to lose power over time and slowly become unable to achieve the maximum desired forward walking speed.  In the other environments the failures were all slipping related.  The algorithm chose the direction and speed which had the highest chance of making the quadruped slip or topple (ramp).  This is what led to the two catastrophic failures you see in the video~\cite{video}.

\section{CONCLUSION}
In this letter, the authors proposed a two-step approach to verification of arbitrary systems subject to an operational specification which has a quantifiable measure of satisfaction.  We show that we can leverage system simulators to accurately lower bound the true-system's capacity to satisfy its specification and also identify the Sim2Real gap between our system simulator and its hardware counterpart.  We demonstrate both repeatability of our results and the ability of our approach to discriminate between different environments in determination of the Sim2Real Gap.  Future work aims to utilize this evaluation approach to iteratively develop better controllers and also minimize risk measures as well.

\renewcommand{\baselinestretch}{0.97}

\bibliographystyle{IEEEtran}
\bibliography{IEEEabrv,bib_works}

\begin{thebibliography}{10}
\providecommand{\url}[1]{#1}
\csname url@rmstyle\endcsname
\providecommand{\newblock}{\relax}
\providecommand{\bibinfo}[2]{#2}
\providecommand\BIBentrySTDinterwordspacing{\spaceskip=0pt\relax}
\providecommand\BIBentryALTinterwordstretchfactor{4}
\providecommand\BIBentryALTinterwordspacing{\spaceskip=\fontdimen2\font plus
\BIBentryALTinterwordstretchfactor\fontdimen3\font minus
  \fontdimen4\font\relax}
\providecommand\BIBforeignlanguage[2]{{%
\expandafter\ifx\csname l@#1\endcsname\relax
\typeout{** WARNING: IEEEtran.bst: No hyphenation pattern has been}%
\typeout{** loaded for the language `#1'. Using the pattern for}%
\typeout{** the default language instead.}%
\else
\language=\csname l@#1\endcsname
\fi
#2}}

\bibitem{video}
\BIBentryALTinterwordspacing
 [Online]. Available: \url{https://youtu.be/uAjEEWIAg3I}
\BIBentrySTDinterwordspacing

\bibitem{kadian2020sim2real}
A.~Kadian, J.~Truong, A.~Gokaslan, A.~Clegg, E.~Wijmans, S.~Lee, M.~Savva,
  S.~Chernova, and D.~Batra, ``Sim2real predictivity: Does evaluation in
  simulation predict real-world performance?'' \emph{IEEE Robotics and
  Automation Letters}, vol.~5, no.~4, pp. 6670--6677, 2020.

\bibitem{tobin2017domain}
J.~Tobin, R.~Fong, A.~Ray, J.~Schneider, W.~Zaremba, and P.~Abbeel, ``Domain
  randomization for transferring deep neural networks from simulation to the
  real world,'' in \emph{2017 IEEE/RSJ international conference on intelligent
  robots and systems (IROS)}.\hskip 1em plus 0.5em minus 0.4em\relax IEEE,
  2017, pp. 23--30.

\bibitem{sadeghi2016cad2rl}
F.~Sadeghi and S.~Levine, ``Cad2rl: Real single-image flight without a single
  real image,'' \emph{arXiv preprint arXiv:1611.04201}, 2016.

\bibitem{andrychowicz2020learning}
O.~M. Andrychowicz, B.~Baker, M.~Chociej, R.~Jozefowicz, B.~McGrew,
  J.~Pachocki, A.~Petron, M.~Plappert, G.~Powell, A.~Ray, \emph{et~al.},
  ``Learning dexterous in-hand manipulation,'' \emph{The International Journal
  of Robotics Research}, vol.~39, no.~1, pp. 3--20, 2020.

\bibitem{matas2018sim}
J.~Matas, S.~James, and A.~J. Davison, ``Sim-to-real reinforcement learning for
  deformable object manipulation,'' in \emph{Conference on Robot
  Learning}.\hskip 1em plus 0.5em minus 0.4em\relax PMLR, 2018, pp. 734--743.

\bibitem{nachum2019multi}
O.~Nachum, M.~Ahn, H.~Ponte, S.~Gu, and V.~Kumar, ``Multi-agent manipulation
  via locomotion using hierarchical sim2real,'' \emph{arXiv preprint
  arXiv:1908.05224}, 2019.

\bibitem{gupta2017cognitive}
S.~Gupta, J.~Davidson, S.~Levine, R.~Sukthankar, and J.~Malik, ``Cognitive
  mapping and planning for visual navigation,'' in \emph{Proceedings of the
  IEEE Conference on Computer Vision and Pattern Recognition}, 2017, pp.
  2616--2625.

\bibitem{dosovitskiy2017carla}
A.~Dosovitskiy, G.~Ros, F.~Codevilla, A.~Lopez, and V.~Koltun, ``Carla: An open
  urban driving simulator,'' in \emph{Conference on robot learning}.\hskip 1em
  plus 0.5em minus 0.4em\relax PMLR, 2017, pp. 1--16.

\bibitem{xia2018gibson}
F.~Xia, A.~R. Zamir, Z.~He, A.~Sax, J.~Malik, and S.~Savarese, ``Gibson env:
  Real-world perception for embodied agents,'' in \emph{Proceedings of the IEEE
  Conference on Computer Vision and Pattern Recognition}, 2018, pp. 9068--9079.

\bibitem{savva2017minos}
M.~Savva, A.~X. Chang, A.~Dosovitskiy, T.~Funkhouser, and V.~Koltun, ``Minos:
  Multimodal indoor simulator for navigation in complex environments,''
  \emph{arXiv preprint arXiv:1712.03931}, 2017.

\bibitem{seshia2016towards}
S.~A. Seshia, D.~Sadigh, and S.~S. Sastry, ``Towards verified artificial
  intelligence,'' \emph{arXiv preprint arXiv:1606.08514}, 2016.

\bibitem{baier2008principles}
C.~Baier and J.-P. Katoen, \emph{Principles of model checking}.\hskip 1em plus
  0.5em minus 0.4em\relax MIT press, 2008.

\bibitem{corso2020survey}
A.~Corso, R.~J. Moss, M.~Koren, R.~Lee, and M.~J. Kochenderfer, ``A survey of
  algorithms for black-box safety validation,'' \emph{arXiv preprint
  arXiv:2005.02979}, 2020.

\bibitem{annpureddy2011s}
Y.~Annpureddy, C.~Liu, G.~Fainekos, and S.~Sankaranarayanan, ``S-taliro: A tool
  for temporal logic falsification for hybrid systems,'' in \emph{International
  Conference on Tools and Algorithms for the Construction and Analysis of
  Systems}.\hskip 1em plus 0.5em minus 0.4em\relax Springer, 2011, pp.
  254--257.

\bibitem{tuncali2016utilizing}
C.~E. Tuncali, T.~P. Pavlic, and G.~Fainekos, ``Utilizing s-taliro as an
  automatic test generation framework for autonomous vehicles,'' in \emph{2016
  ieee 19th international conference on intelligent transportation systems
  (itsc)}.\hskip 1em plus 0.5em minus 0.4em\relax IEEE, 2016, pp. 1470--1475.

\bibitem{donze2010breach}
A.~Donz{\'e}, ``Breach, a toolbox for verification and parameter synthesis of
  hybrid systems,'' in \emph{International Conference on Computer Aided
  Verification}.\hskip 1em plus 0.5em minus 0.4em\relax Springer, 2010, pp.
  167--170.

\bibitem{dreossi2019verifai}
T.~Dreossi, D.~J. Fremont, S.~Ghosh, E.~Kim, H.~Ravanbakhsh,
  M.~Vazquez-Chanlatte, and S.~A. Seshia, ``Verifai: A toolkit for the formal
  design and analysis of artificial intelligence-based systems,'' in
  \emph{International Conference on Computer Aided Verification}.\hskip 1em
  plus 0.5em minus 0.4em\relax Springer, 2019, pp. 432--442.

\bibitem{ghosh2018verifying}
S.~Ghosh, F.~Berkenkamp, G.~Ranade, S.~Qadeer, and A.~Kapoor, ``Verifying
  controllers against adversarial examples with bayesian optimization,'' in
  \emph{2018 IEEE International Conference on Robotics and Automation
  (ICRA)}.\hskip 1em plus 0.5em minus 0.4em\relax IEEE, 2018, pp. 7306--7313.

\bibitem{gangopadhyay2019identification}
B.~Gangopadhyay, S.~Khastgir, S.~Dey, P.~Dasgupta, G.~Montana, and P.~Jennings,
  ``Identification of test cases for automated driving systems using bayesian
  optimization,'' in \emph{2019 IEEE Intelligent Transportation Systems
  Conference (ITSC)}.\hskip 1em plus 0.5em minus 0.4em\relax IEEE, 2019, pp.
  1961--1967.

\bibitem{deshmukh2017testing}
J.~Deshmukh, M.~Horvat, X.~Jin, R.~Majumdar, and V.~S. Prabhu, ``Testing
  cyber-physical systems through bayesian optimization,'' \emph{ACM
  Transactions on Embedded Computing Systems (TECS)}, vol.~16, no.~5s, pp.
  1--18, 2017.

\bibitem{marco2017virtual}
A.~Marco, F.~Berkenkamp, P.~Hennig, A.~P. Schoellig, A.~Krause, S.~Schaal, and
  S.~Trimpe, ``Virtual vs. real: Trading off simulations and physical
  experiments in reinforcement learning with bayesian optimization,'' in
  \emph{2017 IEEE International Conference on Robotics and Automation
  (ICRA)}.\hskip 1em plus 0.5em minus 0.4em\relax IEEE, 2017, pp. 1557--1563.

\bibitem{berkenkamp2016safe}
F.~Berkenkamp, A.~P. Schoellig, and A.~Krause, ``Safe controller optimization
  for quadrotors with gaussian processes,'' in \emph{2016 IEEE International
  Conference on Robotics and Automation (ICRA)}.\hskip 1em plus 0.5em minus
  0.4em\relax IEEE, 2016, pp. 491--496.

\bibitem{berkenkamp2021bayesian}
F.~Berkenkamp, A.~Krause, and A.~P. Schoellig, ``Bayesian optimization with
  safety constraints: safe and automatic parameter tuning in robotics,''
  \emph{Machine Learning}, pp. 1--35, 2021.

\bibitem{wang2013bayesian}
Z.~Wang, M.~Zoghi, F.~Hutter, D.~Matheson, and N.~De~Freitas, ``Bayesian
  optimization in high dimensions via random embeddings,'' in
  \emph{Twenty-Third international joint conference on artificial
  intelligence}, 2013.

\bibitem{rana2017high}
S.~Rana, C.~Li, S.~Gupta, V.~Nguyen, and S.~Venkatesh, ``High dimensional
  bayesian optimization with elastic gaussian process,'' in \emph{International
  conference on machine learning}.\hskip 1em plus 0.5em minus 0.4em\relax PMLR,
  2017, pp. 2883--2891.

\bibitem{rolland2018high}
P.~Rolland, J.~Scarlett, I.~Bogunovic, and V.~Cevher, ``High-dimensional
  bayesian optimization via additive models with overlapping groups,'' in
  \emph{International conference on artificial intelligence and
  statistics}.\hskip 1em plus 0.5em minus 0.4em\relax PMLR, 2018, pp. 298--307.

\bibitem{akella2021learning}
P.~Akella, U.~Rosolia, and A.~D. Ames, ``Learning performance bounds for
  safety-critical systems,'' 2021.

\bibitem{srinivas2009gaussian}
N.~Srinivas, A.~Krause, S.~M. Kakade, and M.~Seeger, ``Gaussian process
  optimization in the bandit setting: No regret and experimental design,''
  \emph{arXiv preprint arXiv:0912.3995}, 2009.

\bibitem{chowdhury2017kernelized}
S.~R. Chowdhury and A.~Gopalan, ``On kernelized multi-armed bandits,'' in
  \emph{International Conference on Machine Learning}.\hskip 1em plus 0.5em
  minus 0.4em\relax PMLR, 2017, pp. 844--853.

\bibitem{bull2011convergence}
A.~D. Bull, ``Convergence rates of efficient global optimization algorithms.''
  \emph{Journal of Machine Learning Research}, vol.~12, no.~10, 2011.

\bibitem{donze2010robust}
A.~Donz{\'e} and O.~Maler, ``Robust satisfaction of temporal logic over
  real-valued signals,'' in \emph{International Conference on Formal Modeling
  and Analysis of Timed Systems}.\hskip 1em plus 0.5em minus 0.4em\relax
  Springer, 2010, pp. 92--106.

\bibitem{ames2016control}
A.~D. Ames, X.~Xu, J.~W. Grizzle, and P.~Tabuada, ``Control barrier function
  based quadratic programs for safety critical systems,'' \emph{IEEE
  Transactions on Automatic Control}, vol.~62, no.~8, pp. 3861--3876, 2016.

\bibitem{wheeler2019critical}
T.~A. Wheeler and M.~J. Kochenderfer, ``Critical factor graph situation
  clusters for accelerated automotive safety validation,'' in \emph{2019 IEEE
  Intelligent Vehicles Symposium (IV)}.\hskip 1em plus 0.5em minus 0.4em\relax
  IEEE, 2019, pp. 2133--2139.

\bibitem{chowdhury2017arxiv}
\BIBentryALTinterwordspacing
S.~R. Chowdhury and A.~Gopalan, ``On kernelized multi-armed bandits,''
  \emph{CoRR}, vol. abs/1704.00445, 2017. [Online]. Available:
  \url{http://arxiv.org/abs/1704.00445}
\BIBentrySTDinterwordspacing

\bibitem{micchelli2006universal}
C.~A. Micchelli, Y.~Xu, and H.~Zhang, ``Universal kernels.'' \emph{Journal of
  Machine Learning Research}, vol.~7, no.~12, 2006.

\bibitem{massart2007concentration}
P.~Massart, \emph{Concentration inequalities and model selection}.\hskip 1em
  plus 0.5em minus 0.4em\relax Springer, 2007.

\bibitem{buchli2009inverse}
J.~Buchli, M.~Kalakrishnan, M.~Mistry, P.~Pastor, and S.~Schaal, ``Compliant
  quadruped locomotion over rough terrain,'' in \emph{IEEE/RSJ International
  Conference on Intelligent Robots and Systems}, 2009, pp. 814--820.

\bibitem{ubellacker2021verifying}
W.~Ubellacker, N.~Csomay-Shanklin, T.~G. Molnar, and A.~D. Ames, ``Verifying
  safe transitions between dynamic motion primitives on legged robots,''
  \emph{arXiv preprint arXiv:2106.10310}, 2021.

\end{thebibliography}

\end{document}